\documentclass[journal,doublecolumn,10 pt]{IEEEtran}

\usepackage{epsfig,latexsym}
\usepackage{float}
\usepackage{indentfirst}
\usepackage{amsmath}
\usepackage{amssymb}
\usepackage{times}
\usepackage{subfigure}
\usepackage{psfrag}
\usepackage{cite}
\usepackage[ruled]{algorithm2e}         

\usepackage{color}

\newtheorem{Proposition}{Proposition}

\begin{document}
\title{Application of Analog Network
Coding to MIMO Two-Way Relay Channel in Cellular
Systems
\thanks{M. Gan and X. Dai are with Department of Electronic Engineering
and Information Science, University of Science and Technology of China,
P.R.China. Z. Ding is with the School of Electrical, Electronic and Computer
Engineering, Newcastle University, UK.
   }}

\author{Ming Gan,   Zhiguo Ding,  \IEEEmembership{Member, IEEE}, Xuchu Dai
\thanks{}}

\twocolumn[
\begin{@twocolumnfalse}
\maketitle
\vspace*{-2 em}
\end{@twocolumnfalse}
]
{
\renewcommand{\thefootnote}{\fnsymbol{footnote}}
\footnotetext {M. Gan and X. Dai are with Department of Electronic Engineering
and Information Science, University of Science and Technology of China,
P.R.China. Z. Ding is with the School of Electrical, Electronic and Computer
Engineering, Newcastle University, UK. }
}


\begin{abstract}
An efficient analog network coding transmission protocol is proposed in
this letter for a MIMO two way cellular network. Block signal alignment is  first  proposed  to  null the  inter-user  interference  for
multi-antenna users, which makes the dimensions of  aligned space larger
compared with the existing signal alignment. Two algorithms are developed to jointly design
the precoding matrices at the relay and BS for outage optimization. Especially, the last algorithm is designed to maximize the effective channel gain to the effective noise gain ratio. The performance of this transmission protocol is also verified by simulations.

\begin{keywords}
MIMO two-way, beamforming, block signal alignment, outage optimization.
\end{keywords}
\end{abstract}

\vspace{-3 ex}

\section{Introduction}
Multiuser-MIMO two-way relaying technique has emerged as a promising technique to improve the spectrum
efficiency by applying multi-antenna techniques to eliminate
inter-user interference \cite{Esli}. To further reduce
the requirement of node antennas, the concept of signal alignment
has been applied to two-way relaying networks \cite{Ding}. Most
existing works about multiuser two-way networks are to design
precoding matrices by completely diagonalizing the channel
matrices. Compared to the concept of block diagonalization, it has
been recognized that perfect diagonalization results in some
performance loss \cite{Spencer}. Therefore, in this letter we will
develop a coordinated beamforming scheme to null the inter-user
interference using a new approach of block signal alignment (BSA),
and derive closed-form expressions for the beamforming matrices
based on orthogonal projection. This new
scheme is completely different from \cite{Ding} whose key idea was point
to point signal alignment (P2PSA) realized by perfect
diagonalization. BSA  makes the dimensions of  aligned space larger
compared with P2PSA such that it achieves better performance both in outage probability and ergodic capacity for multi-antenna
users.

Focusing on outage optimization for each user and applying the
beamforming structures designed at the relay and base station (BS), we first develop
one algorithm to jointly design the beamforming matrices at the
relay and BS based on MMSE and Algebraic Norm-Maximizing (ANOMAX)
\cite{Roemer}. Then we develop another algorithm based on Algebraic Norm-Maximizing without Normalization (ANMwoN), which is first
proposed in this letter to maximize the Frobenius norms of the effective channels without normalizing
beamforming matrices as in ANOMAX. Instead of only maximizing the effective channel gain in ANOMAX, the effective channel gain to the effective
noise gain ratio (ECG2ENGR) is maximized for the relay precoder. Simulation results demonstrate the performance
of the proposed transmission approaches.

\vspace*{-0.25 em}
\section{Description of the Protocol}\label{section description}
Consider a MIMO cellular network with  $M$ mobile users, one BS and a
single relay, where each node operates in half-duplex. Both the
relay and the BS are equipped with $N$ antennas. The $m$-th mobile
user has $K_m$ antennas, for $m = 1,\ldots, M$. The total number of antennas at all
mobile users is defined to be $K=\sum^M_{m=1}{K_m}$. All wireless
channels are assumed to be Rayleigh fading, reciprocal, and
constant over two transmission phases and there are no direct links
between the BS and mobile users due to severe shadowing. It is assumed that the
BS and relay have global channel state information (CSI) prior to
transmission.


Without loss of generality, we focus on the $m$-th user and the similar operations for other users. To achieve full multiplexing gain, it is assumed that each node transmits full data streams, equal to the number of its antennas. During the multiple access phase, the BS transmits the precoded symbols, $\mathbf{P}\mathbf{s}$,
where $\mathbf{s}=\begin{bmatrix} \mathbf{s}^T_1 &\cdots
&\mathbf{s}^T_M
\end{bmatrix}^T$, $\mathbf{s}_m\in\mathbb{C}^{K_m \times
1}$ denotes the message to the $m$-th
user with unit power and $\mathbf{P}$ is a $N\times K$ precoding
matrix. It is assumed that the transmission power for the $m$-th
user at the BS is $\text{P}^{ow}_m$, i.e.,
$\text{tr}\{\mathbf{P}\mathbf{P}^H\}\leq \sum^M_{m=1}
\mathbf{P}^{ow}_m$, where $\text{tr}\{\cdot\}$ denotes the trace.
Meanwhile, the $m$-th users sends its message vector $\mathbf{s'}_m
\in \mathbb{C}^{K_m \times 1}$ with unit power to the relay. Therefore, the observations at the
relay are given as
\begin{eqnarray}\label{1}
\mathbf{r} = \mathbf{G}\mathbf{P}\mathbf{s} + \sum_{m=1}^{M}\mathbf{H}_{m}\mathbf{s'}_m+\mathbf{n}_R,
\end{eqnarray}
where $\mathbf{G}\in \mathbb{C}^{N \times N}$($\mathbf{H}_{m}\in \mathbb{C}^{N \times K_m}$) is the full-rank channel matrix from the BS (the $m$-th user) to the relay,  $\mathbf{n}_R \in\mathbb{C}^{N \times 1}$ denotes additive white Gaussian noise vector whose components follow  $\mathcal{CN}\left(0, \sigma^2_r\right)$. During the broadcast phase, upon receiving $\mathbf{r}$, the relay amplifies it using the precoding matrix $\mathbf{W}$ and broadcasts the output of the precoder. The transmission power constraint for $m$-th user at the relay is $\text{P}^{ow}_{R_m}$.  Then the received signals at the BS  can be written as
\begin{eqnarray}\label{2}
\mathbf{y}_{BS} &=& \mathbf{G}^H\mathbf{W}\left( \mathbf{G}\mathbf{P}\mathbf{s} + \mathbf{H}\mathbf{s'}+\mathbf{n}_R\right)
+\mathbf{n}_{BS},
\end{eqnarray}
where $\mathbf{H}=\begin{bmatrix} \mathbf{H}_{1} &\cdots&\mathbf{H}_{M}\end{bmatrix}$ and $\mathbf{s'}=\begin{bmatrix} \mathbf{s'}^T_{1} &\cdots& \mathbf{s'}^T_{M}\end{bmatrix}^T$, and at the $m$-th user, the signal model can be expressed as
\begin{eqnarray}\label{3}
\mathbf{y}_m = \mathbf{H}_{m}^H\mathbf{W}\left( \mathbf{G}\mathbf{P}\mathbf{s} + \mathbf{H}\mathbf{s'}+\mathbf{n}_R\right) +\mathbf{n}_m,
\end{eqnarray}
where $\mathbf{n}_m\in\mathbb{C}^{K_m \times 1}$ and
$\mathbf{n}_{BS}\in\mathbb{C}^{N \times 1}$ are defined similarly
to $\mathbf{n}_R$, whose components follow $\mathcal{CN}\left(0,
\sigma^2_m\right)$ and $\mathcal{CN}\left(0, \sigma^2_{BS}\right)$,
respectively. Observe that each user receives a mixture of
desirable signals, self-interference (SI), interference
from the messages sent from and to other users and amplified
noise, where SI can be easily removed with the CSI and the last two terms will severely degrade the
performance.

\vspace*{-0.25 em}
\section{The Precoding Design at the BS and Relay}
Here we consider the situation of $N=K$.\footnote{The proposed  approach can be extended to the scenarios with  $N\leq K$ as follows.  The precoding design will be accomplished in two phases. The strategy introduced in Section III and IV will be carried out first, and then a further step is to optimize the user precoding matrices. The second step is analogous to antenna selection in traditional MIMO systems.   Specifically, the message vector  $\mathbf{s'}_m \in \mathbb{C}^{K_m \times 1}$ sent by the $m$-th user denotes $\mathbf{s'}_m=\mathbf{F}_m  \mathbf{x'}_m$, where the precoding matrix $\mathbf{F}_m \in \mathbb{C}^{K_m \times K'_m}$ can be set as the   $K'_m$  right singular vectors of $H_m$ or the end to end effective channel matrix corresponding to the largest singular values, and $\mathbf{x'}_m \in \mathbb{C}^{K'_m \times 1} $ is the coded signal. The parameter $K'=\Sigma^M_{m=1} K'_m$ is chosen to ensure that the effective channel matrix between the relay and the users  $\mathbf{H'}= \begin{bmatrix} \mathbf{H}_1 \mathbf{F}_1 & \cdots & \mathbf{H}_M \mathbf{F}_M \end{bmatrix} \in \mathcal{C}^{N \times K'}$ is non-singular, i.e.  ($N =K'$), which is used instead of $\mathbf{H}$. The reader is referred to \cite{Xin} for a more detailed discussion on the user precoding design.}
Instead of P2PSA proposed in \cite{Ding},  the concept of
block diagonalization is applied for the coordinated beamforming design, and  the key idea of this design is BSA that the observations from
and to the same mobile user can lie in the same space, i.e.,
$\mathbf{s}_m$ and $\mathbf{s'}_m$ aligned together. The dimensions of aligned space become larger since the two coefficient matrices before the two grouped signal vectors are different, not the identity matrices, after the coordinated beamforming design. This can be facilitated
by defining the precoder $\mathbf{P}$ at the BS to meet
\begin{eqnarray}\label{4}
\mathbf{H}^{-1} \mathbf{G}\mathbf{P} = \text{diag} \left( \begin{bmatrix}  \mathbf{A}_1 & \cdots & \mathbf{A}_M\end{bmatrix} \right),
\end{eqnarray}
where $\mathbf{P}=\begin{bmatrix}\mathbf{P}_1 & \mathbf{P}_2 &\cdots&
\mathbf{P}_M \end{bmatrix}$, $\text{diag} \left(\right)$ represents the block diagonal matrix, and $\mathbf{A}_m \in\mathbb{C}^{ K_m \times
K_m}$,  $m=1,2,...,M$ , is a full rank matrix. The use of the precoding matrix defined in \eqref{4} can convert the signal model in \eqref{1} to $\mathbf{r} =  \sum_{m=1}^{M}\mathbf{H}_{m}(\mathbf{A}_m \mathbf{s}_m+\mathbf{s'}_m)+\mathbf{n}_R$ , where the user signals and the signals from the BS are aligned together. The existing scheme in \cite{Ding} is a special case by letting all $\mathbf{A}_m$ to be identity matrices. To satisfy the condition \eqref{4}, the preorder $\mathbf{P}$ can be obtained by using orthogonal projection method.  First, define $\mathbf{G'}=\mathbf{H}^{-1}
\mathbf{G}$ and represent $\mathbf{G'}=\begin{bmatrix}
\mathbf{G'}^T_1 &\cdots& \mathbf{G'}^T_M \end{bmatrix}^T$. Then using
the orthogonal projection matrix
$\mathbf{Q}_m=\left(\mathbf{I}_{N}-\tilde{\mathbf{G'}}^H_m(\tilde{\mathbf{G'}}_m\tilde{\mathbf{G'}}^H_m)^{-1}\tilde{\mathbf{G'}}_m\right)$
where $\tilde{\mathbf{G'}}_m=\begin{bmatrix} \mathbf{G'}^T_1
&\cdots& \mathbf{G'}^T_{m-1} & \mathbf{G'}^T_{m+1} &\cdots&
\mathbf{G'}^T_M \end{bmatrix}^T \in\mathbb{C}^{ \left(K-K_m\right) \times
N} $, the precoder's component matrix $\mathbf{P}_m$ can be expressed as
\vspace*{-0.9 em}
\begin{eqnarray}\label{5}
\mathbf{P}_m = \mathbf{Q}_m \mathbf{G'}^H_m \mathbf{D}_m,
\end{eqnarray}
where the fact $\mathbf{G'}_i \mathbf{Q}_j= \mathbf{0}$ for $i\neq j$ is used and $\mathbf{D}_m$ is used to meet the power
constraint for the $m$-th user at the BS  as
\begin{eqnarray}\label{6}
\text{P}^{ow}_m = \text{tr}\left\{ \mathbf{P}_m \mathbf{P}^H_m \right\}=\text{tr}\left\{ \mathbf{T}_m \mathbf{D}_m \mathbf{D}^H_m \right\},
\end{eqnarray}
where $\mathbf{T}_m=\mathbf{G'}_m \mathbf{Q}_m \mathbf {G'}^H_m$, a $K_m \times
K_m$ matrix.
With the Hermitian matrix assumption for the power constraint matrix $\mathbf{D}_m$ in \eqref{6}, it can be computed via
\begin{eqnarray}\label{7}
\mathbf{D}_m  = \frac{\text{P}^{ow}_m}{\sqrt{K_m}} \left( \mathbf{T}_m \right)^{-\frac{1}{2}},
\end{eqnarray}
which yields a deterministic solution for the precoding matrices.
Obviously such a fixed solution is not optimal, and how to optimize
precoding will be discussed in the next section.
Substituting \eqref{4} into \eqref{3}, we obtain $\mathbf{y}_m = \mathbf{H}_{m}^H\mathbf{W}\left(\mathbf{H}\left(\text{diag}\left(\begin{bmatrix} \mathbf{T}_1 \mathbf{D}_1 &\cdots&  \mathbf{T}_M \mathbf{D}_M \end{bmatrix}\right) \mathbf{s} + \mathbf{s'}\right)+\mathbf{n}_R\right)+\mathbf{n}_m$. To remove the inter-user interference, the precoding matrix $\mathbf{W}$ at the relay should follow
\begin{eqnarray}\label{8}
\mathbf{H}^H\mathbf{W}\mathbf{H} = \text{diag} \left(\begin{bmatrix}  \mathbf{B}_1 & \cdots & \mathbf{B}_M\end{bmatrix} \right),
\end{eqnarray}
where $\mathbf{W}=\sum^M_{m=1}\mathbf{W}_m$, $\mathbf{B}_m$ is a full rank $K_m \times K_m$ matrix. Note that the condition \eqref{8} is not perfect diagonalization, enhancing the BSA. It can be realized by the combination of the orthogonal projection and eigen value decomposition (EVD). First define $\tilde{\mathbf{H}}_m=\begin{bmatrix} \mathbf{H}_1 &\cdots& \mathbf{H}_{m-1} & \mathbf{H}_{m+1} &\cdots& \mathbf{H}_{M}\end{bmatrix} \in \mathbb{C}^{ N\times \left(K-K_m \right)}$ with its orthogonal projection matrix is $\mathbf{Q'}_m= \left(\mathbf{I}_N-\tilde{\mathbf{H}}_m(\tilde{\mathbf{H}}^H_m\tilde{\mathbf{H}}_m)^{-1}\tilde{\mathbf{H}}^H_m\right)$ of idempotent property. Then using EVD, i.e., $\mathbf{Q'}_m = \begin{bmatrix} \mathbf{U}_{m_1} & \mathbf{U}_{m_0}\end{bmatrix} \begin{bmatrix} \mathbf{\Sigma}_m & 0 \\ 0 & 0 \end{bmatrix} {\begin{bmatrix} \mathbf{U}_{m_1} & \mathbf{U}_{m_0}\end{bmatrix}}^H $, where $\mathbf{U}_{m_1} \in \mathbb{C}^{ N \times K_m }$, the precoder's component matrix $\mathbf{W}_m$ can be computed via
\begin{eqnarray}\label{9}
\mathbf{W}_m=\mathbf{U}_{m_1} \mathbf{D'}_m \mathbf{U}^H_{m_1},
\end{eqnarray}
where the fact $\mathbf{U}^H_{i_1}\mathbf{H}_j =\mathbf{0}$ for $i\neq j$ is used and $\mathbf{D'}_m$ is to ensure the power constraint for messages transmission to the $m$-th user at the relay as
\begin{eqnarray}\label{10}
\text{P}^{ow}_{R_m} &=& tr\left\{ \mathbf{W}_m \mathbf{r} \mathbf{r}^H \mathbf{W}^H_m \right\}
\\&\approx&
tr\left\{  \mathbf{D'}^H_m  \mathbf{D'}_m \bar{\mathbf{T}}_m \left(\mathbf{T}_m \mathbf{D}_m \mathbf{D}^H_m \mathbf{T}^H_m +\mathbf{I}_{K_m}\right)\bar{\mathbf{T}}^H_m \right\}, \nonumber
\end{eqnarray}
where $\bar{\mathbf{T}}_m = \mathbf{U}^H_{m_1} \mathbf{H}_m$ and the last  approximation is due to the high signal-to-noise ratio (SNR)
assumption. With the Hermitian matrix assumption for
$\mathbf{D'}_m $, a deterministic solution of \eqref{10} can be expressed as
\begin{eqnarray}\label{11}
\mathbf{D'}_m  = \frac{\text{P}^{ow}_{R_m}}{\sqrt{K_m}} \left( \bar{\mathbf{T}}_m \left(\mathbf{T}_m \mathbf{D}_m \mathbf{D}^H_m \mathbf{T}^H_m +\mathbf{I}_{K_m}\right)\bar{\mathbf{T}}^H_m \right)^{-\frac{1}{2}}.
\end{eqnarray}
In the second phase, the relay nulls the inter-user interference using the fact $\mathbf{W}_i \mathbf{H}_j = \mathbf{0}$ and $\mathbf{H}^H_i \mathbf{W}_j = \mathbf{0}$, for $i\neq j $. Hence \eqref{3} can be rewritten as
\begin{eqnarray}\label{12}
\mathbf{y}_m & = & \mathbf{B}_m  \mathbf{A}_m \mathbf{s}_m + \mathbf{B}_m  \mathbf{s'}_m + \mathbf{H}^H_m \mathbf{W}_m  \mathbf{n}_R +\mathbf{n}_m
\\&=& \bar{\mathbf{T}}^H_m  \mathbf{D'}_m \bar{\mathbf{T}}_m  \mathbf{T}_m \mathbf{D}_m \mathbf{s}_m + \bar{\mathbf{T}}^H_m  \mathbf{D'}_m  \bar{\mathbf{T}}_m  \mathbf{s'}_m + \tilde{\mathbf{n}}_m,\nonumber
\end{eqnarray}
where  $\mathbf{B}_m=\bar{\mathbf{T}}^H_m  \mathbf{D'}_m
\bar{\mathbf{T}}_m $, $\mathbf{A}_m=\mathbf{T}_m \mathbf{D}_m$ and
$\tilde{\mathbf{n}}_m = \bar{\mathbf{T}}^H_m \mathbf{D'}_m
\mathbf{U}^H_{m_1} \mathbf{n}_R +\mathbf{n}_m$. As seen from
\eqref{12}, the $m$-th user only observes a mixture of the desired
signals, SI and enhanced noise, where the signals from and to other
users have been eliminated completely because of the careful
precoding design  at the BS and relay. After the cancellation of SI
, the mutual information at the $m$-th user and BS can be
calculated by
\begin{eqnarray}\label{13}
\mathcal{I}_m = \text{log} \text{det}\left( \mathbf{I}_{K_m} + \mathbf{F}_m\mathbf{F}^H_m \tilde{\mathbf{R}}^{-1}_m\right),
\end{eqnarray}
where $\mathbf{F}_m= \bar{\mathbf{T}}^H_m  \mathbf{D'}_m \bar{\mathbf{T}}_m  \mathbf{T}_m \mathbf{D}_m $ and $\tilde{\mathbf{R}}_m=\bar{\mathbf{T}}^H_m \mathbf{D'}_m \mathbf{D'}^H_m \bar{\mathbf{T}}_m \sigma^2_r + \sigma^2_m \mathbf{I}_{K_m}$ .

\vspace*{-0.5 em}
\section{Outage Optimization for each user}
Due to the poor capability of users who are equipped with less
antennas, we focus on the precoding design optimization for users.
Based on the beamforming structures designed
in Section III, we jointly determine the optimal power
normalization matrices $\mathbf{D}_m$ and $\mathbf{D'}_m$ for the $m$-th user's outage optimization without Hermitian
matrices assumption, i.e., $\min\limits_{\mathbf{D}_m,\mathbf{D'}_m} \text{Pro}\left(\mathcal{I}_m < 2 K_m R \right)$, where $R$ is the target data rate per
channel use. However, this problem is extremely complicated,
non-linear and non-convex and hence difficult to be solved. To obtain the approximate optimal solution,
we propose two iterative algorithms based on alternating optimization that updates one precoder at a time while fixing another.
It is assumed that the power constraints for the $m$-th user at both BS and relay are 1.
\vspace*{-1 em}
\subsection{Outage Optimization based on MMSE and ANOMAX}
We use the combination of the MMSE algorithm and ANOMAX algorithm \cite{Roemer} to obtain the optimal solution iteratively. For given $\mathbf{D'}_m$, using MMSE algorithm \cite{Joham} and based on  the equation \eqref{12}, we obtain
\begin{eqnarray}\label{14}
\mathbf{D}_m = \sqrt{\gamma_m} \mathbf{C}_m  ,
\end{eqnarray}
where $\mathbf{C}_m = \left(\tilde{\mathbf{F}}^H_m \tilde{\mathbf{F}}_m + \text{tr} \left(\tilde{\mathbf{R}}_m\right)\mathbf{I}_{K_m} \right)^{-1} \tilde{\mathbf{F}}^H_m$, $\tilde{\mathbf{F}}_m=\bar{\mathbf{T}}^H_m \mathbf{D'}_m
\bar{\mathbf{T}}_m \mathbf{T}_m$ and $\gamma_m $ is to ensure the power constraint equation \eqref{6}, i.e., $\gamma_m = 1 /\text{tr}\left(\mathbf{T}_m \mathbf{C}_m \mathbf{C}^H_m\right)$.

However, MMSE is not optimal relay precoding design strategy in the two-way relaying scenario as illustrated in \cite{Roemer}. Hence for given $\mathbf{D}_m$, we use ANOMAX algorithm to optimize the relay precoder, focusing on the effective channel gain in \eqref{12}. Therefore,
the optimization problem is expressed as
\begin{eqnarray}\label{15}
\max\limits_{\mathbf{D'}_m}  & \left\|\bar{\mathbf{T}}^H_m  \mathbf{D'}_m \bar{\mathbf{T}}_m \mathbf{T}_m \mathbf{D}_m \right\|^2_{F} \nonumber
\\ s.t. & \mathbf{P}^{ow}_{R_m}\leq 1, \; m=1,2,...,M.
\end{eqnarray}
We assume $\mathbf{D'}_m=\sqrt{\gamma'_m} \mathbf{C'}_m$  with normalization of $\|\mathbf{C'}_m\|^2_{F}=1 $  as in \cite{Roemer}, and to satisfy the power constraint, substituting it into the equation \eqref{10} yields
\begin{eqnarray}\label{16}
\gamma'_m = 1 /\text{tr}\left(\mathbf{C'}_m  \bar{\mathbf{T}}_m \left(\mathbf{T}_m \mathbf{D}_m \mathbf{D}^H_m  \mathbf{T}^H_m +\mathbf{I}_{K_m}\right)\bar{\mathbf{T}}^H_m \mathbf{C'}^H_m\right).
\end{eqnarray}
Using the fact $\|\mathbf{X}\|_F=\|\text{vec}\{\mathbf{X}\}\|_2$ and $\text{vec}\{\mathbf{X}\mathbf{Z}\mathbf{Y}\}=\left(\mathbf{Y}^T \bigotimes \mathbf{X}\right)\text{vec}\{\mathbf{Z}\}$ as in \cite{Roemer}, the objective function in \eqref{15} can be simplified as $\mathbf{c'}^H_m \mathbf{K}^H_m \mathbf{K}_m \mathbf{c'}_m$, ignoring the constant factor $\gamma'_m $, where $\mathbf{K}_m=\left(\bar{\mathbf{T}}_m  \mathbf{T}_m \mathbf{D}_m \right)^T \bigotimes \bar{\mathbf{T}}^H_m$, $\mathbf{c'}_m=\text{vec}\left(\mathbf{C'}_m\right)$ and $\text{vec}(\cdot)$ operator stacks the column vectors of the matrix into a column vector. Then it can be rewritten as $\frac{\mathbf{c'}^H_m \mathbf{K}^H_m \mathbf{K}_m \mathbf{c'}_m}{\mathbf{c'}^H_m \mathbf{c'}_m}$ due to the normalization. Hence, the maximum value of this function is the square of the largest singular value of $\mathbf{K}_m$ and
\begin{eqnarray}\label{17}
\mathbf{C'}_m = \text{unvec}(\mathbf{u}_m),
\end{eqnarray}
where $\mathbf{u}_m$ is the dominant right singular vector of matrix $\mathbf{K}_m$ corresponding to the largest singular value and $\text{unvec}(\cdot)$ denotes the opposite operation of $\text{vec}(\cdot) $.  In summary, we propose a solution as in \textbf{Algorithm} \ref{iterative algorithm 1}.
\vspace*{-1 em}
\begin{algorithm}
\label{iterative algorithm 1}
\caption{based on MMSE and ANOMAX}
$\cdot$ \textbf{Initialize} power normalized matrix $\mathbf{D}_m$ using \eqref{7}\;
$\cdot$ \textbf {Repeat}

\quad $-$ {\text{Compute} $\mathbf{D'}_m$ using  \eqref{16} and \eqref{17} \;
\quad $-$  \text{Compute} $\mathbf{D}_m$ using  \eqref{14} \;
\quad $-$  \text{Compute} $\mathbf{D'}_m$ \text{again} and \text{compute} $\mathcal{I}_m$ using \eqref{13} \;}
$\cdot$ \textbf {Until } {$\mathcal{I}_m$ converges.}

\end{algorithm}
\vspace*{-0.25 em}
\subsection{Outage Optimization based on ANMwoN}
Observe that the optimal relay precoder in \cite{Roemer} should
be  scaled up to satisfy the constraint with equality while
increasing the objective function, which contradicts the
optimality. Here we propose the ANMwoN strategy to solve this problem.

First, we solve the optimization problem for given $\mathbf{D'}_m$. To apply the ANMwoN strategy
to  the precoding design at the BS,  we extend
the beamforming structure $\mathbf{P}_m$ at the BS to be
\begin{eqnarray}\label{18}
\bar{\mathbf{P}}_m = \bar{\mathbf{U}}_{m_1} \bar{\mathbf{D}}_m
\bar{\mathbf{U}}^H_{m_1} \mathbf{G'}^H_m,
\end{eqnarray}
where $\bar{\mathbf{U}}_{m_1}\in\mathbb{C}^{N \times K_m}$ is obtained by the EVD of $\mathbf{Q}_m$, i.e., $\mathbf{Q}_m =
\begin{bmatrix} \bar{\mathbf{U}}_{m_1} & \bar{\mathbf{U}}_{m_0}\end{bmatrix} \begin{bmatrix}
\bar{\mathbf{\Sigma}}_m & 0 \\ 0 & 0 \end{bmatrix}
{\begin{bmatrix} \bar{\mathbf{U}}_{m_1} &
\bar{\mathbf{U}}_{m_0}\end{bmatrix}}^H $. Substituting $\bar{\mathbf{P}}_m$ into \eqref{6} instead of $\mathbf{P}_m$, i.e., $\mathbf{P}^{ow}_m=\text{tr}\left\{ \bar{\mathbf{D}}^H_m  \bar{\mathbf{D}}_m  \mathbb{T}^H_m \mathbb{T}_m \right\}$,  with Hermitian
matrix assumption for the initial power normalized matrix $\bar{\mathbf{D}}_m$, it can be
expressed as
\begin{eqnarray}\label{19}
\bar{\mathbf{D}}_m = \frac{1}{\sqrt{K_m}} \left( \mathbb{T}^H_m \mathbb{T}_m\right)^{-\frac{1}{2}},
\end{eqnarray}
where $\mathbb{T}_m=\mathbf{G'}_m \bar{\mathbf{U}}_{m_1}$, and the power constraint $\mathbf{P}^{ow}_m$ can be rewritten as $\mathbf{P}^{ow}_m=\bar{\mathbf{d}}^H_m  \bar{\mathbf{R}}_m \bar{\mathbf{d}}_m \leq 1$, where $\bar{\mathbf{R}}_m= \left( \mathbb{T}^H_m \mathbb{T}_m \right)^T  \bigotimes \mathbf{I}_{K_m} $ and $\bar{\mathbf{d}}_m=\text{vec}\left(\bar{\mathbf{D}}_m\right)$. Since the effective noise term is not related to $\mathbf{D}_m$ directly in \eqref{12}, we concentrate on the effective channel gain as \eqref{15} using $\mathbb{T}_m \bar{\mathbf{D}}_m \mathbb{T}^H_m $  instead of $\mathbf{T}_m \mathbf{D}_m$ as follows
\begin{eqnarray}\label{20}
\left\|\bar{\mathbf{T}}^H_m  \mathbf{D'}_m \bar{\mathbf{T}}_m
\mathbb{T}_m \bar{\mathbf{D}}_m \mathbb{T}^H_m \right\|^2_{F}.
\end{eqnarray}
 It can be simplified as $\bar{\mathbf{d}}^H_m \bar{\mathbf{K}}^H_m \bar{\mathbf{K}}_m \bar{\mathbf{d}}_m$ where $\bar{\mathbf{K}}_m= \left(\mathbb{T}^H_m \right)^T \bigotimes \left(\bar{\mathbf{T}}^H_m
\mathbf{D'}_m \bar{\mathbf{T}}_m \mathbb{T}^H_m \right)$. Obviously observe that the optimal vector $\tilde{\mathbf{d}}_m$ for the solution $\bar{\mathbf{d}}_m$ of such objective function  must satisfy the power constraint with equality. Hence we write this optimization problem for given $\mathbf{D'}_m$ as
\begin{eqnarray}\label{21}
 \max\limits_{\bar{\mathbf{d}}_m} & \bar{\mathbf{d}}^H_m \bar{\mathbf{K}}^H_m \bar{\mathbf{K}}_m \bar{\mathbf{d}}_m  \nonumber
\\
 s.t. & \bar{\mathbf{d}}^H_m  \bar{\mathbf{R}}_m \bar{\mathbf{d}}_m=1.
\end{eqnarray}
Then we put power constraint into the objective function as
\begin{eqnarray}\label{22}
\max\limits_{\bar{\mathbf{d}}_m} \; \frac{\bar{\mathbf{d}}^H_m \bar{\mathbf{K}}^H_m \bar{\mathbf{K}}_m \bar{\mathbf{d}}_m}{\bar{\mathbf{d}}^H_m  \bar{\mathbf{R}}_m \bar{\mathbf{d}}_m}.
\end{eqnarray}
The problem \eqref{22} is equal to \eqref{21} since the objective function is homogeneous and any scaling in $\bar{\mathbf{d}}_m$ does not change the optimality. We refer to this method as ANMwoN strategy.
\begin{Proposition}\label{Proposition2}
The optimal vector $\tilde{\mathbf{d}}_m$ of \eqref{22} is the dominant eigenvector of $\bar{\mathbf{R}}^{-1}_m \bar{\mathbf{K}}^H_m \bar{\mathbf{K}}_m$.
\end{Proposition}
\begin{proof}
Rewrite \eqref{22} as $\lambda\left(\bar{\mathbf{d}}_m\right)= \frac{\bar{\mathbf{d}}^H_m \bar{\mathbf{K}}^H_m \bar{\mathbf{K}}_m \bar{\mathbf{d}}_m}{\bar{\mathbf{d}}^H_m  \bar{\mathbf{R}}_m \bar{\mathbf{d}}_m} $. All local optimal vectors $\bar{\mathbf{d}}_m$ satisfy the first order necessary Karush-Kuhn-Tucker (KKT) condition as $ \frac{\partial \lambda \left(\bar{\mathbf{d}}\right)} {\partial \bar{\mathbf{d}}^{*}} = 0$, where $\left(\cdot\right)^\ast$ represents complex conjugation. Then we simplify KKT condition equation as $\bar{\mathbf{K}}^H_m \bar{\mathbf{K}}_m \bar{\mathbf{d}}_m = \lambda \left(\bar{\mathbf{d}}_m\right) \bar{\mathbf{R}}_m \bar{\mathbf{d}}_m$. Hence the optimal vector $\tilde{\mathbf{d}}_m$ is the dominant eigenvector
of the matrix $\bar{\mathbf{R}}^{-1}_m \bar{\mathbf{K}}^H_m \bar{\mathbf{K}}_m$  due to its uniqueness.
\end{proof}
We obtain the optimal $\bar{\mathbf{D}}_m$ satisfied the power constraint as
\begin{eqnarray}\label{23}
\bar{\mathbf{D}}_m =\text{unvec} \left( \frac{1}{\tilde{\mathbf{d}}^H_m  \bar{\mathbf{R}}_m \tilde{\mathbf{d}}}_m \tilde{\mathbf{d}}_m \right).
\end{eqnarray}

Second, we solve the optimization problem for given $\mathbf{D}_m$.  Since the effective noise term is related to $\mathbf{D'}_m$ directly in \eqref{12}, we
consider another objective function including the noise term. Since
the signal streams and noise streams at each user are correlated, it
is difficult to analyse them. Inspired by the concept of the
effective channel gain, we propose a relaxing objective function
based on the ECG2ENGR as
\begin{eqnarray}\label{24}
\frac {\left\|\bar{\mathbf{T}}^H_m  \mathbf{D'}_m \bar{\mathbf{T}}_m \mathbb{T}_m \bar{\mathbf{D}}_m \mathbb{T}^H_m  \right\|^2_{F}} {\text{tr}\left(\bar{\mathbf{T}}^H_m \mathbf{D'}_m \mathbf{D'}^H_m \bar{\mathbf{T}}_m \sigma^2_r +  \sigma^2_m \mathbf{I}_{K_m}\right)}.
\end{eqnarray}
Using ANMwoN method, we put the relay power constraint $\mathbf{P}^{ow}_{R_m}= \mathbf{d'}^H_m  \bar{\mathbf{R'}}_m \mathbf{d'}_m=1$ into the second term of denominator in \eqref{24}, where $\mathbf{d'}_m=\text{vec}\left(\mathbf{D'}_m\right)$ and $\bar{\mathbf{R'}}_m=\left( \bar{\mathbf{T}}_m \left(\mathbb{T}_m \mathbf{D}_m  \mathbb{T}^H_m  \mathbb{T}_m \mathbf{D}^H_m  \mathbb{T}^H_m +\mathbf{I}_{K_m}\right)\bar{\mathbf{T}}^H_m \right)^T \bigotimes \mathbf{I}_{K_m}$. Then the equation \eqref{24} can be simplified as $\frac{\mathbf{d'}^H_m\mathbb{K}^H_m \mathbb{K}_m \mathbf{d'}_m}{\mathbf{d'}^H_m  \tilde{\mathbf{K}}_m \mathbf{d'}_m}$, where $\mathbb{K}_m= \left( \bar{\mathbf{T}}_m \mathbb{T}_m \bar{\mathbf{D}}_m \mathbb{T}^H_m \right)^T \bigotimes \bar{\mathbf{T}}^H_m$ and $\tilde{\mathbf{K}}_m = \sigma^2_r \left( \mathbf{I}_{K_m} \bigotimes \left(\bar{\mathbf{T}}_m \bar{\mathbf{T}}^H_m\right) \right) + \sigma^2_m \bar{\mathbf{R'}}_m$. Hence the optimal matrix $\mathbf{D'}_m $ can be expressed as
\begin{eqnarray}\label{25}
\mathbf{D'}_m =\text{unvec} \left( \frac{1}{\tilde{\mathbf{d'}}^H_m  \bar{\mathbf{R'}}_m \tilde{\mathbf{d'}}_m} \tilde{\mathbf{d'}}_m\right),
\end{eqnarray}
where $\tilde{\mathbf{d'}}_m$ is the dominant eigenvector of $\tilde{\mathbf{K}}^{-1}_m \mathbb{K}^H_m \mathbb{K}_m$. In summary, we develop an algorithm based on ANMwoN and ECG2ENGR  as in \textbf{Algorithm} \ref{iterative algorithm 2}.
\vspace*{-1 em}
\begin{algorithm}
\label{iterative algorithm 2}
\caption{based on ANMwoN and ECG2ENGR}
$\cdot$ \textbf{Initialize} power normalized matrix $\bar{\mathbf{D}}_m$ using \eqref{19}\;
$\cdot$ \textbf {Repeat}

\quad $-$ {\text{Compute} $\mathbf{D'}_m$ using  \eqref{25}  \;
\quad $-$  \text{Compute} $\bar{\mathbf{D}}_m$ using  \eqref{23} \;
\quad $-$  \text{Compute} $\mathbf{D'}_m$ \text{again} and \text{compute} $\mathcal{I}_m$ using \eqref{13} \\
\quad with $\bar{\mathbf{F}}_m= \bar{\mathbf{T}}^H_m  \mathbf{D'}_m \bar{\mathbf{T}}_m \mathbb{T}_m \bar{\mathbf{D}}_m \mathbb{T}^H_m $\ instead of $\mathbf{F}_m$ \;}
$\cdot$ \textbf {Until } {$\mathcal{I}_m$ converges.}

\end{algorithm}
\vspace*{-2 em}

\section{Simulation results}
In this section,  we will perform simulations with the antenna
configuration as $N=4$, $M=2$ and $K_m=2$, $m=1,2$ and the
transmission data rate $R=1$ bit per
channel use (BPCU).   In our simulations, we set
$\sigma_m = \sigma_r = \sigma_{BS}=\sigma$ such that $\text{SNR}=
\frac{1}{\sigma^2}$ and consider the Rayleigh fading channels, the components of all channel matrices follow $\mathcal{CN}\left(0, 1\right)$. The proposed schemes are compared with P2PSA scheme
\cite{Ding}, and  the scheme based on the time sharing analog
network coding scheme,   where each user exchanges information
with BS in a round robin way, which needs two time slots for one
round robin.

In Fig. \ref{fig2},  we compare the outage probability
of users for the time sharing scheme, P2PSA scheme and BSA schemes.
We observe that the proposed BSA schemes can achieve better outage
performance than the schemes based on time sharing and P2PSA.
Since P2PSA is based on perfect diagonalization, its outage performance is
the worst. Meanwhile we see that BSA schemes with iterative
algorithms performance better than the BSA scheme with predefined precoders,
and the BSA scheme with algorithm 2 is the best among all BSA schemes since
it is based on ANMwoN and ECG2ENGR. In Fig. \ref{fig3}, we compare
their ergodic capacity performance. As shown from the figure, the
schemes based on BSA and P2PSA  are superior to the time sharing
scheme because the latter needs many time slots to accomplish
the whole transmission. BSA schemes yield a rates improvement
compared with the P2PSA scheme, due to the larger aligned space. Note that the BSA scheme with predefined precoders yields higher rates than BSA schemes with algorithm 1 and 2, because these two algorithms are designed to outage optimization for users, not optimal for maximizing the sum-rate of the whole system and they tend to concentrate most of the energy on the dominant singular value, not beneficial to the sum-rate over high SNRs.

\begin{figure}[!t]\centering
    \epsfig{file=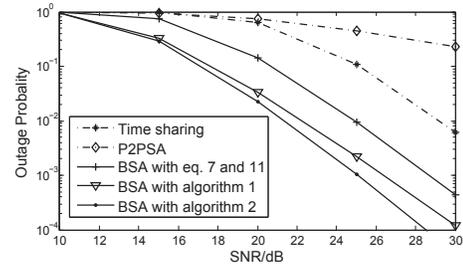, width=0.35\textwidth,clip=}
    \vspace*{-3mm}
    \caption{Outage Probability for users, $Pro(I_m<2 K_m R)$, $R=1$ BPCU, $N=4$, $M=2$ and $K_m=2$, $m=1,2$.  }
    \label{fig2}
    \vspace*{-1mm}
\end{figure}

\begin{figure}[!t]\centering
    \epsfig{file=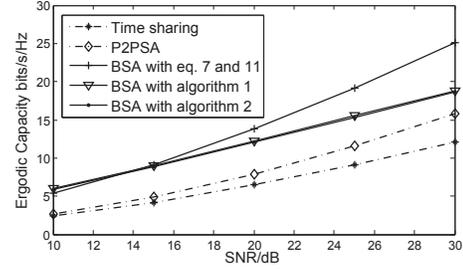, width=0.35\textwidth,clip=}
    \vspace*{-3mm}
    \caption{Ergodic Capacity,  $N=4$, $M=2$ and $K_m=2$, $m=1,2$.  }
    \label{fig3}
    \vspace*{-2mm}
\end{figure}

\section{Conclusion}
We have  developed a coordinated beamforming scheme
to null the inter-user interference for a MIMO cellular network
using a approach of BSA. Moreover, we have also  proposed two
algorithms to jointly optimize the beamforming matrices at the
relay and BS for outage optimization.  Simulation results have
been provided for performance evaluation.

\end{document}